\newif\ifarxiv
\arxivtrue 

\ifarxiv
\documentclass[a4paper]{article}
\else
\documentclass[a4paper,envcountsame]{llncs}
\fi
\pdfoutput=1

\usepackage[T1]{fontenc}
\usepackage[utf8]{inputenc} 
\usepackage[american]{babel}

\usepackage[bookmarks=false]{hyperref}
\usepackage[final]{microtype}

\ifarxiv
\usepackage[a4paper,centering]{geometry}
\fi

\usepackage{amsmath}
\usepackage{amssymb}

\ifarxiv\else

\fi
\usepackage{amsthm}
\ifarxiv
\theoremstyle{plain}
\newtheorem{theorem}{Theorem}
\newtheorem{lemma}[theorem]{Lemma}
\newtheorem{proposition}[theorem]{Proposition}
\newtheorem{corollary}[theorem]{Corollary}
\theoremstyle{definition}
\newtheorem{definition}[theorem]{Definition}
\fi

\mathchardef\mhyph="2D

\newcommand{\binalph}{\{ 0, 1 \}}
\newcommand{\eps}{\varepsilon}
\newcommand{\N}{\mathbb{N}}

\newcommand{\Z}{\mathbb{Z}}

\let\SS\relax
\DeclareMathOperator{\bin}{bin}
\DeclareMathOperator{\SS}{SS}

\usepackage{complexity}
\newclass{\ACA}{ACA}
\newclass{\CA}{CA}
\newclass{\DLOGTIME}{DLOGTIME}
\newclass{\LT}{LT}
\newclass{\DACA}{DACA}
\newclass{\SLT}{SLT}

\newlang{\BIN}{BIN}
\newlang{\DELTA}{DELTA}
\newlang{\IDMAT}{IDMAT}
\newlang{\SOMEONE}{SOMEONE}
\newlang{\STATE}{STATE}

\newcommand{\FOL}{\FO_\L}

\usepackage[misc]{ifsym} 
\usepackage{pifont}
\newcommand{\cmark}{\ding{51}}
\newcommand{\xmark}{\ding{55}}

\usepackage{csquotes}

\usepackage{tikz}
\usetikzlibrary{arrows,automata,calc,positioning}

\usepackage[backend=biber,bibencoding=utf8,sortcites=true,maxnames=2]{biblatex}
\AtEveryBibitem{%
  \clearfield{doi}
  \clearfield{isbn}
  \clearfield{issn}
  \clearfield{note}
  \clearfield{url}
}

\newcommand{\myomit}[1]{%
  \ifarxiv#1\fi
}

\title{Sublinear-Time Language Recognition and Decision by One-Dimensional
Cellular Automata%
  \ifarxiv\else\thanks{Some proofs have been omitted due to page constraints.
    These may be found in the full version of the paper
 \cite{modanese19_sublinear-time_arxiv}.}\fi
}
\newcommand{\myname}{Augusto Modanese}
\newcommand{\myaffil}{Karlsruhe Institute of Technology (KIT), Germany}
\newcommand{\myemail}{modanese@kit.edu}
\ifarxiv
\author{\myname \\ \\ \normalsize{\myaffil} \\ \normalsize{\texttt{\myemail}}}
\date{}
\else
\author{\myname$^{(\textrm{\Letter})}$}
\institute{Karlsruhe Institute of Technology (KIT), Germany \\
  \email{modanese@kit.edu}
}
\fi

\begin{document}

\maketitle

\begin{abstract}
  After an apparent hiatus of roughly 30 years, we revisit a seemingly neglected
subject in the theory of (one-dimensional) cellular automata: sublinear-time
computation.
  The model considered is that of ACAs, which are language acceptors whose
acceptance condition depends on the states of all cells in the automaton.
  We prove a time hierarchy theorem for sublinear-time ACA classes, analyze
their intersection with the regular languages, and, finally, establish strict
inclusions in the parallel computation classes $\SC$ and (uniform) $\AC$.
  As an addendum, we introduce and investigate the concept of a decider ACA
(DACA) as a candidate for a decider counterpart to (acceptor) ACAs.
  We show the class of languages decidable in constant time by DACAs equals the
locally testable languages, and we also determine $\Omega(\sqrt{n})$ as the
(tight) time complexity threshold for DACAs up to which no advantage compared
to constant time is possible.
\end{abstract}


\section{Introduction}

While there have been several works on linear- and real-time language
recognition by cellular automata over the years (see, e.g.,
\cite{kutrib_language_theory, terrier12} for an overview), interest in the
sublinear-time case has been scanty at best.
We can only speculate this has been due to a certain obstinacy concerning what
is now the established acceptance condition for cellular automata, namely that
the first cell determines the automaton's response, despite alternatives being
long known \cite{rosenfeld_book}.
Under this condition, only a constant-size prefix can ever influence the
automaton's decision, which effectively dooms sublinear time to be but a trivial
case just as it is for (classical) Turing machines, for example.
Nevertheless, at least in the realm of Turing machines, this shortcoming was
readily circumvented by adding a random access mechanism to the model, thus
sparking rich theories on parallel computation \cite{cook_NC, ruzzo_NC_ATMs},
probabilistically checkable proofs \cite{sudan_PCP_survey}, and property testing
\cite{fischer_property_testing_survey, rubinfeld_property_testing_survey}.

In the case of cellular automata, the adaptation needed is an alternate (and by
all means novel) acceptance condition, covered in Section~\ref{sec_def}.
Interestingly, in the resulting model, called ACA, parallelism and local
behavior seem to be more marked features, taking priority over cell
communication and synchronization algorithms (which are the dominant themes in
the linear- and real-time constructions).
As mentioned above, the body of theory on sublinear-time ACAs is very small and,
to the best of our knowledge, resumes itself to \cite{Ibarra_ACA, kim_ACA,
sommerhalder_ACA}.
\citeauthor{Ibarra_ACA} \cite{Ibarra_ACA} show sublinear-time ACAs are capable
of recognizing non-regular languages and also determine a threshold (namely
$\Omega(\log n)$) up to which no advantage compared to constant time is
possible.
Meanwhile, \citeauthor{kim_ACA} \cite{kim_ACA} and \citeauthor{sommerhalder_ACA}
\cite{sommerhalder_ACA} analyze the constant-time case subject to different
acceptance conditions and characterize it based on the locally testable
languages, a subclass of the regular languages.

Indeed, as covered in Section~\ref{sec_first_obs}, the defining property of the
locally testable languages, that is, that words which locally appear to be the
same are equivalent with respect to membership in the language at hand,
effectively translates into an inherent property of acceptance by sublinear-time
ACAs.
In Section~\ref{sec_main_results}, we prove a time hierarchy theorem for
sublinear-time ACAs as well as further relate the language classes they define
to the regular languages and the parallel computation classes $\SC$ and
(uniform) $\AC$.
In the same section, we also obtain an improvement on a result of
\cite{Ibarra_ACA}.
Finally, in Section~\ref{sec_DACA} we consider a plausible model of ACAs as
language deciders, that is, machines which must not only accept words in the
target language but also explicitly reject those which do not.
Section~\ref{sec_conclusion} concludes.

\section{Definitions} \label{sec_def}

We assume the reader is familiar with the theory of formal languages and
cellular automata as well as with computational complexity theory (see, e.g.,
standard references \cite{arora09_computational_book, delorme99_cellular_book}).
This section reviews basic concepts and introduces ACAs.

$\Z$ denotes the set of integers, $\N_+$ that of (strictly) positive integers,
and $\N_0 = \N_+ \cup \{ 0 \}$.
\myomit{For a function $f\colon A \to B$ and $A' \subseteq A$, $f|A'$ indicates
  the restriction of $f$ to $A'$.}
$B^A$ is the set of functions $f\colon A \to B$.
For a word $w \in \Sigma^\ast$ over an alphabet $\Sigma$, $w(i)$ is the $i$-th
symbol of $w$ (starting with the $0$-th symbol), and $|w|_x$ is the number of
occurrences of $x \in \Sigma$ in $w$.
For $k \in \N_0$, $p_k(w)$, $s_k(w)$, and $I_k(w)$ are the prefix, suffix
and set of infixes of length $k$ of $w$, respectively, where $p_{k'}(w) =
s_{k'}(w) = w$ and $I_{k'}(w) = \{ w \}$ for $k' \ge |w|$.
$\Sigma^{\le k}$ is the set of words $w \in \Sigma^\ast$ for which $|w| \le k$.
Unless otherwise noted, $n$ is the input length.

\subsection{(Strictly) Locally Testable Languages}

The class $\REG$ of regular languages is defined in terms of (deterministic)
automata with finite memory and which read their input in a single direction
(i.e., from left to right), one symbol at a time; once all symbols have been
read, the machine outputs a single bit representing its decision.
In contrast, a \emph{scanner} is a memoryless machine which reads a span of $k
\in \N_+$ symbols at a time of an input provided with start and end markers (so
it can handle prefixes and suffixes separately); the scanner validates every
such substring it reads using the same predicate, and it accepts if and only if
all these validations are successful.
The languages accepted by these machines are the strictly locally testable
languages.%
\footnote{The term \enquote{(locally) testable in the strict sense} ((L)TSS) is
also common \cite{kim_ACA,mcnaughton_LT,mcnaughton_papert}.}

\begin{definition}[strictly locally testable] \label{def_SLT}
  Let $\Sigma$ be an alphabet.
  A language $L \subseteq \Sigma^\ast$ is \emph{strictly locally testable} if
there is some $k \in \N_+$ and sets $\pi, \sigma \subseteq \Sigma^{\le k}$ and
$\mu \subseteq \Sigma^k$ such that, for every word $w \in \Sigma^\ast$, $w \in
L$ if and only if $p_k(w) \in \pi$, $I_k(w) \subseteq \mu$, and $s_k(w) \in
\sigma$.
  $\SLT$ is the class of strictly locally testable languages.
\end{definition}

A more general notion of locality is provided by the locally testable languages.
Intuitively, $L$ is locally testable if a word $w$ being in $L$ or not is
entirely dependent on a property of the substrings of $w$ of some constant
length $k \in \N_+$ (that depends only on $L$, not on $w$).
Thus, if any two words have the same set of substrings of length $k$, then they
are equivalent with respect to being in $L$:

\begin{definition}[locally testable] \label{def_LT}
  Let $\Sigma$ be an alphabet.
  A language $L \subseteq \Sigma^\ast$ is \emph{locally testable} if there is
some $k \in \N_+$ such that, for every $w_1, w_2 \in \Sigma^\ast$ with $p_k(w_1)
= p_k(w_2)$, $I_k(w_1) = I_k(w_2)$, and $s_k(w_1) = s_k(w_2)$ we have that $w_1
\in L$ if and only if $w_2 \in L$.
  $\LT$ denotes the class of locally testable languages.
\end{definition}

$\LT$ is the \emph{Boolean closure} of $\SLT$, that is, its closure under
union, intersection, and complement \cite{mcnaughton_papert}.
In particular, $\SLT \subsetneq \LT$ (i.e., the inclusion is proper
\cite{mcnaughton_LT}).

\subsection{Cellular Automata} \label{sec_def_ACA}

In this paper, we are strictly interested in one-dimensional cellular automata
with the standard neighborhood.
For $r \in \N_0$, let $N_r(z) = \{ z' \in \Z \mid |z - z'| \le r \}$ denote the
\emph{extended neighborhood of radius $r$} of the cell $z \in \Z$.

\begin{definition}[cellular automaton] \label{def_CA}
  A \emph{cellular automaton} (\emph{CA}) $C$ is a triple $(Q, \delta, \Sigma)$
where $Q$ is a finite, non-empty set of \emph{states}, $\delta\colon Q^3 \to Q$
is the \emph{local transition function}, and $\Sigma \subseteq Q$ is the
\emph{input alphabet}.
  An element of $Q^3$ (resp., $Q^\Z$) is called a \emph{local} (resp.,
\emph{global}) \emph{configuration} of $C$.
  $\delta$ induces the \emph{global transition function} $\Delta\colon Q^\Z \to
Q^\Z$ on the configuration space $Q^\Z$ by $\Delta(c)(z) = \delta(c(z-1), c(z),
c(z+1))$, where $z \in \Z$ is a cell and $c \in Q^\Z$.
\end{definition}

Our interest in CAs is as machines which receive an input and process it until a
final state is reached.
The input is provided from left to right, with one cell for each input symbol.
The surrounding cells are inactive and remain so for the entirety of the
computation (i.e., the CA is bounded).
It is customary for CAs to have a distinguished cell, usually cell zero, which
communicates the machine's output.
As mentioned in the introduction, this convention is inadequate for computation
in sublinear time; instead, we require the finality condition to depend on the
entire (global) configuration (modulo inactive cells):

\begin{definition}[CA computation]
  There is a distinguished state $q \in Q \setminus \Sigma$, called the
\emph{inactive state}, which, for every $z_1, z_2, z_3 \in Q$, satisfies
$\delta(z_1, z_2, z_3) = q$ if and only if $z_2 = q$.
  A cell not in state $q$ is said to be \emph{active}.
  For an input $w \in \Sigma^\ast$, the \emph{initial configuration} $c_0 =
c_0(w) \in Q^\Z$ of $C$ for $w$ is $c_0(i) = w(i)$ for $i \in \{ 0, \dots, |w| -
1 \}$ and $c_0(i) = q$ otherwise.
  For $F \subseteq Q \setminus \{ q \}$, a configuration $c \in Q^\Z$ is
$F$-\emph{final} (for $w$) if there is a (minimal) $\tau \in \N_0$ such that $c
= \Delta^\tau(c_0)$ and $c$ contains only states in $F \cup \{ q \}$.
  In this context, the sequence $c_0, \dots, \Delta^\tau(c_0) = c$ is the
\emph{trace} of $w$, and $\tau$ is the \emph{time complexity} of $C$ (with
respect to $F$ and $w$).
\end{definition}

Because we effectively consider only bounded CAs, the computation of $w$
involves exactly $|w|$ active cells.
The surrounding inactive cells are needed only as markers for the start and end
of $w$.
As a side effect, the initial configuration $c_0 = c_0(\eps)$ for the empty word
$\eps$ is stationary (i.e., $\Delta(c_0) = c_0$) regardless of the choice of
$\delta$.
Since this is the case only for $\eps$, we disregard it for the rest of the
paper, that is, we assume it is not contained in any of the languages
considered.

Finally, we relate final configurations and computation results.
We adopt an acceptance condition as in \cite{sommerhalder_ACA, rosenfeld_book}
and obtain a so-called ACA; here, the \enquote{A} of \enquote{ACA} refers to the
property that all (active) cells are relevant for acceptance.

\begin{definition}[ACA] \label{def_ACA}
  An \emph{ACA} is a CA $C$ with a non-empty subset $A \subseteq Q \setminus \{
q \}$ of \emph{accept states}.
  For $w \in \Sigma^+$, if $C$ reaches an $A$-final configuration, we say $C$
\emph{accepts} $w$.
  $L(C)$ denotes the set of words accepted by $C$.
  For $t\colon \N_+ \to \N_0$, we write $\ACA(t)$ for the class of languages
accepted by an ACA with time complexity bounded by $t$, that is, for which the
time complexity of accepting $w$ is $\le t(|w|)$.
\end{definition}

$\ACA(t_1) \subseteq \ACA(t_2)$ is immediate for functions $t_1, t_2\colon \N_+
\to \N_0$ with $t_1(n) \le t_2(n)$ for every $n \in \N_+$.
Because Definition~\ref{def_ACA} allows multiple accept states, it is possible
for each (non-accepting) state $z$ to have a corresponding accept state $z_A$.
In the rest of this paper, when we say a cell becomes (or marks itself as)
accepting (without explicitly mentioning its state), we intend to say it changes
from such a state $z$ to $z_A$.

Figure~\ref{fig_ACA_0ast} illustrates the computation of an ACA with input
alphabet $\Sigma = \binalph$ and which accepts $\{ 01 \}^+$ with time complexity
equal to one (step).
The local transition function is such that $\delta(0,1,0) = \delta(1,0,1) =
\delta(q,0,1) = \delta(0,1,q) = a$, $a$ being the (only) accept state, and
$\delta(z_1, z_2, z_3) = z_2$ for $z_2 \neq a$ and arbitrary $z_1$ and $z_3$.

\begin{figure}[t]
  \centering
  \begin{tikzpicture}[every node/.style={block},
      block/.style={
        minimum width=width("$0$") + 3mm,
        minimum height=height("$0$") + 3mm,
        text height=1.5ex,
        text depth=.3ex,
        outer sep=0pt,
        draw, rectangle, node distance=0pt}]

    \newlength{\rowspc}
    \setlength{\rowspc}{1.55em}
    \newlength{\colspc}
    \setlength{\colspc}{19.65em}

    \node    (z0)                            {$q$};
    \node    (z1)  [right=of z0]             {$0$};
    \node    (z2)  [right=of z1]             {$1$};
    \node    (z3)  [right=of z2]             {$0$};
    \node    (z4)  [right=of z3]             {$1$};
    \node    (z5)  [right=of z4]             {$0$};
    \node    (z6)  [right=of z5]             {$1$};
    \node    (z7)  [right=of z6]             {$q$};
    \draw    (z0.north west) -- ++(-0.4cm,0)
             (z0.south west) -- ++(-0.4cm,0)
             (z7.north east) -- ++(0.4cm,0)
             (z7.south east) -- ++(0.4cm,0);

    \node    (x0)  [yshift=-\rowspc]          {$q$};
    \node    (x1)  [right=of x0]             {$a$};
    \node    (x2)  [right=of x1]             {$a$};
    \node    (x3)  [right=of x2]             {$a$};
    \node    (x4)  [right=of x3]             {$a$};
    \node    (x5)  [right=of x4]             {$a$};
    \node    (x6)  [right=of x5]             {$a$};
    \node    (x7)  [right=of x6]             {$q$};
    \draw    (x0.north west) -- ++(-0.4cm,0)
             (x0.south west) -- ++(-0.4cm,0)
             (x7.north east) -- ++(0.4cm,0)
             (x7.south east) -- ++(0.4cm,0);

    \node    (w0)  [xshift=\colspc]          {$q$};
    \node    (w1)  [right=of w0]             {$0$};
    \node    (w2)  [right=of w1]             {$0$};
    \node    (w3)  [right=of w2]             {$1$};
    \node    (w4)  [right=of w3]             {$0$};
    \node    (w5)  [right=of w4]             {$1$};
    \node    (w6)  [right=of w5]             {$0$};
    \node    (w7)  [right=of w6]             {$q$};
    \draw    (w0.north west) -- ++(-0.4cm,0)
             (w0.south west) -- ++(-0.4cm,0)
             (w7.north east) -- ++(0.4cm,0)
             (w7.south east) -- ++(0.4cm,0);

    \node    (y0)  [xshift=\colspc,yshift=-\rowspc] {$q$};
    \node    (y1)  [right=of y0]             {$0$};
    \node    (y2)  [right=of y1]             {$0$};
    \node    (y3)  [right=of y2]             {$a$};
    \node    (y4)  [right=of y3]             {$a$};
    \node    (y5)  [right=of y4]             {$a$};
    \node    (y6)  [right=of y5]             {$0$};
    \node    (y7)  [right=of y6]             {$q$};
    \draw    (y0.north west) -- ++(-0.4cm,0)
             (y0.south west) -- ++(-0.4cm,0)
             (y7.north east) -- ++(0.4cm,0)
             (y7.south east) -- ++(0.4cm,0);

    \node (cm) [draw=none, right=of x7, xshift=6mm] {\cmark};

  \end{tikzpicture}
  \caption{Computation of an ACA which recognizes $L = \{ 01 \}^+$.
    The input words are $010101 \in L$ and $001010 \not\in L$, respectively.}
  \label{fig_ACA_0ast}
\end{figure}

\section{First Observations} \label{sec_first_obs}

This section recalls results on sublinear-time ACA computation (i.e., $\ACA(t)$
where $t \in o(n)$) from \cite{kim_ACA, sommerhalder_ACA, Ibarra_ACA} and
provides some additional remarks.
We start with the constant-time case (i.e., $\ACA(O(1))$).
Here, the connection between scanners and ACAs is apparent:
If an ACA accepts an input $w$ in time $\tau = \tau(w)$, then $w$ can be
verified by a scanner with an input span of $2\tau + 1$ symbols and using the
predicate induced by the local transition function of the ACA (i.e., the
predicate is true if and only if the symbols read correspond to $N_\tau(z)$ for
some cell $z$ in the initial configuration and $z$ is accepting after $\tau$
steps).

Constant-time ACA computation has been studied in
\cite{kim_ACA,sommerhalder_ACA}.
Although in \cite{kim_ACA} we find a characterization based on a hierarchy over
$\SLT$, the acceptance condition there differs slightly from that in
Definition~\ref{def_ACA}; in particular, the automata there run for a number of
steps which is fixed for each automaton, and the outcome is evaluated (only) in
the final step.
In contrast, in \cite{sommerhalder_ACA} we find the following, where $\SLT_\lor$
denotes the closure of $\SLT$ under union:

\begin{theorem}[\cite{sommerhalder_ACA}] \label{thm_ACA_SLT}
  $\ACA(O(1)) = \SLT_\lor$.
\end{theorem}

Thus, $\ACA(O(1))$ is closed under union.
In fact, more generally:

\begin{proposition} \label{prop_waca_union}
  For any $t\colon \N_+ \to \N_+$, $\ACA(O(t))$ is closed under union.
\end{proposition}

\myomit{
  \begin{proof}
    Let $L_1$ and $L_2$ be languages accepted by the ACAs $C_1$ and $C_2$,
    respectively, in $O(t)$ time.
    Furthermore, let $Q_i$ (resp., $Q_i^a$) denote the set of states (resp.,
    accepting states) of $C_i$.
    We construct an ACA $C$ which accepts $L_1 \cup L_2$ as follows:
    $C$ simulates $C_1$ and $C_2$ and switches between the two simulations at
    every step.
    Each cell maintains components $q_1 \in Q_1$, $q_2 \in Q_2$, and $r \in \{
    1,
    2 \}$, where $r$ indicates which of the two simulations is to be updated
    next;
    that is, at each step, a cell next updates $q_r$ according to the local
    transition function of $C_r$ and the $q_r$ states of its neighbors.
    At the start of the computation, all cells simultaneously initialize $r =
    1$.
    The accept states of $C$ are $Q_1 \times Q_2^a \times \{ 1 \} \cup Q_1^a
    \times Q_2 \times \{ 2 \}$.
    Thus, $L(C) = L_1 \cup L_2$ and the time complexity of $C$ is in $O(2t) =
    O(t)$.
  \end{proof}
}

$\ACA(O(1))$ is closed under intersection \cite{sommerhalder_ACA}.
It is an open question whether $\ACA(O(t))$ is also closed under intersection
for every $t \in o(n)$.
\myomit{In particular, note that taking the construction of the proof above
  and setting, for example, $Q_1^a \times Q_2^a \times \{ 1 \}$ as the set of
  accept states is insufficient (as it is not guaranteed that $C_1$ and $C_2$
  accept \emph{at the same time}).}

Moving beyond constant time, in \cite{Ibarra_ACA} we find the following:

\begin{theorem}[\cite{Ibarra_ACA}] \label{thm_ACA_log}
  For $t \in o(\log n)$, $\ACA(t) \subseteq \REG$.
\end{theorem}

In \cite{Ibarra_ACA} we find an example for a non-regular language in
$\ACA(O(\log n))$ which is essentially a variation of the language
\[
\BIN = \{ \bin_k(0) \# \bin_k(1) \# \cdots \# \bin_k(2^k - 1) \mid k \in \N_+
\}
\]
where $\bin_k(m)$ is the $k$-digit binary representation of $m \in \{ 0,
\dots, 2^k - 1 \}$.%

To illustrate the ideas involved, we present an example related to $\BIN$
(though it results in a different time complexity) and which is also useful in
later discussions in Section~\ref{sec_DACA}.
Let $w_k(i) = 0^i 1 0^{k-i-1}$ and consider the language
\[
\IDMAT = \{ w_k(0) \# w_k(1) \# \cdots \# w_{k}(k-1) \mid k \in \N_+ \}
\]
of all identity matrices in line-for-line representations, where the lines are
separated by $\#$ symbols.%
\footnote{Alternatively, one can also think of $\IDMAT$ as a (natural) problem
  on graphs presented in the adjacency matrix representation.}

We now describe an ACA for $\IDMAT$; the construction closely follows the
aforementioned one for $\BIN$ found in \cite{Ibarra_ACA} (and the difference in
complexity is only due to the different number and size of blocks in the words
of $\IDMAT$ and $\BIN$).
Denote each group of cells initially containing a (maximally long) $\binalph^+$
substring of $w \in \IDMAT$ by a \emph{block}.
Each block of size $b$ propagates its contents to the neighboring blocks (in
separate registers); using a textbook CA technique, this requires exactly $2b$
steps.
Once the strings align, a block initially containing $w_k(i)$ verifies it has
received $w_k(i-1)$ and $w_k(i+1)$ from its left and right neighbor blocks (if
either exists), respectively.
The cells of a block and its delimiters become accepting if and only if the
comparisons are successful and there is a single $\#$ between the block and its
neighbors.
This process takes linear time in $b$; since any $w \in \IDMAT$ has
$O(\sqrt{|w|})$ many blocks, each with $b \in O(\sqrt{|w|})$ cells, it follows
that $\IDMAT \in \ACA(O(\sqrt{n}))$.

To show the above construction is time-optimal, we use the following
observation, which is also central in proving several other results in this
paper:

\begin{lemma} \label{lem_locality_ACA}
  Let $C$ be an ACA, and let $w$ be an input which $C$ accepts in (exactly)
  $\tau = \tau(w)$ steps.
  Then, for every input $w'$ such that $p_{2\tau}(w) = p_{2\tau}(w')$, $I_{2\tau
    + 1}(w') \subseteq I_{2\tau + 1}(w)$, and $s_{2\tau}(w) = s_{2\tau}(w')$,
  $C$ accepts $w'$ in at most $\tau$ steps.
\end{lemma}

The lemma is intended to be used with $\tau < \frac{|w|}{2}$ since otherwise $w
= w'$.
It can be used, for instance, to show that $\SOMEONE = \{ w \in \binalph^+ \mid
|w|_1 \ge 1 \}$ is not in $\ACA(t)$ for any $t \in o(n)$ (e.g., set $w = 0^k 1
0^k$ and $w' = 0^{2k+1}$ for large $k \in \N_+$).
It follows $\REG \not\subseteq \ACA(t)$ for $t \in o(n)$.

\myomit{
  \begin{proof}
    By definition, $C$ does not accept $w'$ in strictly less than $\tau$ steps.
    Let $A$ be the set of accept states of $C$, and let $c_0$ and $c_0'$ denote
    the initial configurations for $w$ and $w'$, respectively.
    Given $w \in L(C)$, we prove $c_\tau' = \Delta^\tau(c_0')$ is $A$-final.

    Let $i \in \Z$.
    If all of $c_0'|N_\tau(i)$ is inactive, then cell $i$ is also inactive in
    $c_\tau'$ (i.e., $c_\tau'(i) = q$).
    If $c_0'|N_\tau(i)$ contains both inactive and active states, then $i <
    \tau$
    or $i \ge |w'| - \tau$, in which case $p_{2\tau}(w) = p_{2\tau}(w')$ and
    $s_{2\tau}(w) = s_{2\tau}(w')$ imply $c_0'|N_\tau(i) = c_0|N_\tau(i)$.
    Finally, if $c_0'|N_\tau(i)$ is purely active, $c_0'|N_\tau(i)$ (seen as a
    word over the input alphabet of $C$) is in $I_{2\tau+1}(w') \subseteq
    I_{2\tau+1}(w)$; as a result, there is $j \in \Z$ with $\tau \le j < |w| -
    \tau$ such that $c_0'|N_\tau(i) = c_0|N_\tau(j)$.
    In both the previous cases, because $c_\tau$ is $A$-final, it follows
    $c_\tau'(i) \in A$, implying $c_\tau'$ is $A$-final.
  \end{proof}
}

Since the complement of $\SOMEONE$ (respective to $\binalph^+$) is $\{0\}^+$
and $\{0\}^+ \in \ACA(O(1))$ (e.g., simply set $0$ as the ACA's accepting
state), $\ACA(t)$ is not closed under complement for any $t \in o(n)$.
Also, $\SOMEONE$ is a regular language and $\BIN \in \ACA(O(\log n))$ is not,
so we have:

\begin{proposition} \label{prop_ACA_REG}
  For $t \in \Omega(\log n) \cap o(n)$, $\ACA(t)$ and $\REG$ are incomparable.
\end{proposition}

If the inclusion of infixes in Lemma~\ref{lem_locality_ACA} is strengthened to
an equality, one may apply it in both directions and obtain the following
stronger statement:

\begin{lemma} \label{lem_locality_ACA_iff}
  Let $C$ be an ACA with time complexity bounded by $t\colon \N_+ \to \N_0$
  (i.e., $C$ accepts any input of length $n$ in at most $t(n)$ steps).
  Then, for any two inputs $w$ and $w'$ with $p_{2\mu}(w) = p_{2\mu}(w')$,
  $I_{2\mu + 1}(w) = I_{2\mu + 1}(w')$, and $s_{2\mu}(w) = s_{2\mu}(w')$ where
  $\mu = \max\{t(|w|), t(|w'|)\}$, we have that $w \in L(C)$ if and only if $w'
  \in L(C)$.
\end{lemma}

\myomit{
\begin{proof}
  The key idea is that, for any $k \in \N_0$, $p_k(w) = p_k(w')$ implies
  $p_{k'}(w) =  p_{k'}(w')$ for every $k' \le k$; the same is true for $s_k(w)
  = s_k(w')$ and $I_k(w) = I_k(w')$.
  Hence, if $w \in L(C)$ and $C$ accepts $w$ in exactly $\tau = \tau(w)$ steps,
  by Lemma~\ref{lem_locality_ACA} we have that $C$ accepts $w'$ in at most $\tau
  \le \mu$ many steps.
  Conversely, if $w' \in L(C)$ and $C$ accepts $w'$ in exactly $\tau = \tau(w')$
  steps, then using Lemma~\ref{lem_locality_ACA} again (with $w$ in place of
  $w'$ and vice-versa) we obtain that $C$ accepts $w$ in at most $\tau \le \mu$
  many steps.
\end{proof}
}

Finally, we can show our ACA for $\IDMAT$ is time-optimal:

\begin{proposition}
  For any $t \in o(\sqrt{n})$, $\IDMAT \not\in \ACA(t)$.
\end{proposition}

\myomit{
\begin{proof}
  Let $n \in \N_+$ be such that $t(n) < \frac{1}{8}\sqrt{n}$, and let $w \in
\IDMAT$ be given with $|w| = n$.
  In particular, $|w| = k^2 + k - 1$ for some $k \in \N_+$ and $w = w_k(0) \#
\cdots \# w_k(k-1)$; without restriction, we also assume $k$ is even.
  Let $w'$ be the word obtained from $w$ by replacing its ($\frac{k}{2} - 1$)-th
block with $w_k(\frac{k}{2})$.
  Then, $p_{2t(n)}(w) = p_{2t(n)}(w')$, $s_{2t(n)}(w) = s_{2t(n)}(w')$, and
$I_{2t(n) + 1}(w) \subseteq I_{2t(n) + 1}(w')$, where the latter follows
from $w_k(\frac{k}{2}) = 0^{k / 2} 1 0^{k / 2 - 1}$, $|w_k(i)| = k$, and $t(n) <
\frac{1}{8}(k + 1)$.
  Thus, by Lemma~\ref{lem_locality_ACA}, if an ACA with time complexity bounded
by $t$ accepts $w$, so does it accept $w' \not\in \IDMAT$.
\end{proof}
}

\section{Main Results} \label{sec_main_results}

In this section, we present various results regarding $\ACA(t)$ where $t \in
o(n)$.
First, we obtain a time hierarchy theorem, that is, under plausible conditions,
$\ACA(t') \subsetneq \ACA(t)$ for $t' \in o(t)$.
Next, we show $\ACA(t) \cap \REG$ is (strictly) contained in $\LT$ and also
present an improvement to Theorem~\ref{thm_ACA_log}.
Finally, we study inclusion relations between $\ACA(t)$ and the $\SC$ and
(uniform) $\AC$ hierarchies.
Save for the material covered so far, all three subsections stand out
independently from one another.

\subsection{Time Hierarchy} \label{sec_ACA_time_hier}

For functions $f,t\colon \N_+ \to \N_0$, we say $f$ is \emph{time-constructible
by CAs in $t(n)$ time} if there is a CA $C$ which, on input $1^n$, reaches a
configuration containing the value $f(n)$ (binary-encoded) in at most $t(n)$
steps.%
\footnote{Just as is the case for Turing machines, there is not a single
definition for time-constructibility by CAs (see, e.g., \cite{CA_counters} for
an alternative).
  Here, we opt for a plausible variant which has the benefit of simplifying the
ensuing line of argument.}
Note that, since CAs can simulate (one-tape) Turing machines in real-time,
any function constructible by Turing machines (in the corresponding sense)
is also constructible by CAs.

\begin{theorem} \label{thm_ACA_time_hierarchy}
  Let $f \in \omega(n)$ with $f(n) \le 2^n$, $g(n) = 2^{n - \lfloor \log f(n)
\rfloor}$, and let $f$ and $g$ be time-constructible (by CAs) in $f(n)$ time.
  Furthermore, let $t\colon \N_+ \to \N_0$ be such that $3 f(k) \le t(f(k)g(k))
\le c f(k)$ for some constant $c \ge 3$ and all but finitely many $k \in \N_+$.
  Then, for every $t' \in o(t)$, $\ACA(t') \subsetneq \ACA(t)$.
\end{theorem}

Given $a > 1$, this can be used, for instance, with any time-constructible $f\in
\Theta(n^a)$ (resp., $f\in \Theta(2^{n/a})$, in which case $a = 1$ is also
possible) and $t \in \Theta((\log n)^a)$ (resp., $t \in \Theta(n^{1/a})$).
The proof idea is to construct a language $L$ similar to $\BIN$ (see
Section~\ref{sec_first_obs}) in which every $w \in L$ has length exponential in
the size of its blocks while the distance between any two blocks is
$\Theta(t(|w|))$.
Due to Lemma~\ref{lem_locality_ACA}, the latter implies $L$ is not recognizable
in $o(t(|w|))$ time.

\begin{proof}
  For simplicity, let $f(n) > n$.
  Consider $L = \{ w_k \mid k \in \N_+ \}$ where
  \[
    w_k = \bin_k(0) \#^{f(k) - k}
      \bin_k(1) \#^{f(k) - k}
      \cdots \bin_k(g(k) - 1) \#^{f(k) - k}
  \]
  and note $|w_k| = f(k) g(k)$.
  Because $t(|w_k|) \in O(f(k))$ and $f(k) \in \omega(k)$, given any $t' \in
o(t)$, setting $w = w_k$, $w' = 0^k \#^{|w_k| - k}$, and $\tau = t'(|w_k|)$ and
applying Lemma~\ref{lem_locality_ACA} for sufficiently large $k$ yields $L
\not\in \ACA(t')$.

  By assumption it suffices to show $w = w_k \in L$ is accepted by an ACA $C$
in at most $3 f(k) \le t(|w|)$ steps for sufficiently large $k \in \N_+$.
  The cells of $C$ perform two procedures $P_1$ and $P_2$ simultaneously:
  $P_1$ is as in the ACA for $\BIN$ (see Section~\ref{sec_first_obs}) and
ensures that the blocks of $w$ have the same length, that the respective binary
encodings are valid, and that the last value is correct (i.e., equal to $g(k) -
1$).
  In $P_2$, each block computes $f(k)$ as a function of its block length $k$.
  Subsequently, the value $f(k)$ is decreased using a real-time counter (see,
e.g., \cite{CA_counters} for a construction).
  Every time the counter is decremented, a signal starts from the block's
leftmost cell and is propagated to the right.
  This allows every group of cells of the form $bs$ with $b \in \{ 0,1 \}^+$
and $s \in \{ \# \}^+$ to assert there are precisely $f(k)$ symbols in total
(i.e., $|bs| = f(k)$).
  A cell is accepting if and only if it is accepting both in $P_1$ and $P_2$.
  The proof is complete by noticing either procedure takes a maximum of $3f(k)$
steps (again, for sufficiently large $k$).
\end{proof}

\subsection{Intersection with the Regular Languages}

In light of Proposition~\ref{prop_ACA_REG}, we now consider the intersection
$\ACA(t) \cap \REG$ for $t \in o(n)$ (in the same spirit as a conjecture by
Straubing \cite{straubing_book}).
For this section, we assume the reader is familiar with the theory of syntactic
semigroups (see, e.g., \cite{Eilenberg_Vol_B} for an in-depth treatment).

Given a language $L$, let $\SS(L)$ denote the syntactic semigroup of $L$.
It is well-known that $\SS(L)$ is finite if and only if $L$ is regular.
A semigroup $S$ is a \emph{semilattice} if $x^2 = x$ and $xy = yx$ for
every $x,y \in S$.
Additionally, $S$ is \emph{locally semilattice} if $eSe$ is a semilattice for
every \emph{idempotent} $e \in S$, that is, $e^2 = e$.
We use the following characterization of locally testable languages:

\begin{theorem}[\cite{brzozowski_LT, mcnaughton_LT}] \label{thm_mcnaughton}
  $L \in \LT$ if and only if $\SS(L)$ is finite and locally semilattice.
\end{theorem}

In conjunction with Lemma~\ref{lem_locality_ACA}, this yields the following,
where the strict inclusion is due to $\SOMEONE \not\in \ACA(t)$ (since $\SOMEONE
\in \LT$; see Section~\ref{sec_first_obs}):

\begin{theorem} \label{thm_ACA_cap_reg}
  For every $t \in o(n)$, $\ACA(t) \cap \REG \subsetneq \LT$.
\end{theorem}

\begin{proof}
  Let $L \in \ACA(t)$ be a language over the alphabet $\Sigma$ and, in addition,
let $L \in \REG$, that is, $S = \SS(L)$ is finite.
  By Theorem~\ref{thm_mcnaughton}, it suffices to show $S$ is locally
semilattice.
  To that end, let $e \in S$ be idempotent, and let $x, y \in S$.

  To show $(exe)(eye) = (eye)(exe)$, let $a, b \in \Sigma^\ast$ and consider
the words $u = a (exe) (eye) b$ and $v = a (eye) (exe) b$.
  For $m \in \N_+$, let $u_m' = a (e^m x e^m) (e^m y e^m) b$, and let $r \in
\N_+$ be such that $r > \max\{ |x|, |y|, |a|, |b| \}$ and also $t(|u_{2r+1}'|) <
\frac{1}{16|e|} |u_{2r+1}'| < r$.
  Since $e$ is idempotent, $u' = u_{2r+1}'$ and $u$ belong to the same class in
$S$, that is, $u' \in L$ if and only if $u \in L$; the same is true for $v' = a
(e^{2r + 1} y e^{2r + 1}) (e^{2r + 1} x e^{2r + 1}) b$ and $v$.
  Furthermore, $p_{2r}(u') = p_{2r}(v')$, $I_{2r + 1}(u') = I_{2r + 1}(v')$, and
$s_{2r}(u') = s_{2r}(v')$ hold.
  Since $L \in \ACA(t)$, Lemma~\ref{lem_locality_ACA_iff} applies.

  The proof of $(exe)(exe) = exe$ is analogous.
  Simply consider the words $a (e^m x e^m) b$ and $a (e^m x e^m) (e^m x e^m) b$
for sufficiently large $m \in \N_+$ and use, again,
Lemma~\ref{lem_locality_ACA_iff} and the fact that $e$ is idempotent.
\end{proof}

Using Theorems~\ref{thm_ACA_log}~and~\ref{thm_ACA_cap_reg}, we have $\ACA(t)
\subsetneq \LT$ for $t \in o(\log n)$.
We can improve this bound to $\ACA(O(1)) = \SLT_\lor$, which is a proper subset
of $\LT$:

\begin{theorem} \label{thm_ACA_olog}
  For every $t \in o(\log n)$, $\ACA(t) = \ACA(O(1))$.
\end{theorem}

\begin{proof}
  We prove every ACA $C$ with time complexity at most $t \in o(\log n)$ actually
  has $O(1)$ time complexity.
  Let $Q$ be the state set of $C$ and assume $|Q| \ge 2$, and let $n_0 \in \N_+$
  be such that $t(n) < \frac{\log n}{9 \log |Q|}$ for $n \ge n_0$.
  Letting $k(n) = 2t(n) + 1$ and assuming $t(n) \ge 1$, we then have
  $|Q|^{3k(n)} \le |Q|^{9t(n)} < n$ ($\star$).
  We shall use this to prove that, for any word $w \in L$ of length $|w| \ge
  n_0$, there is a word $w' \in L$ of length $|w'| \le n_0$ as well as $r < n_0$
  such that $p_r(w) = p_r(w')$, $I_{r+1}(w) = I_{r+1}(w')$, and $s_r(w) =
  s_r(w')$.
  By Lemma~\ref{lem_locality_ACA}, $C$ must have $t(|w'|)$ time complexity
  on $w$ and, since the set of all such $w'$ is finite, it follows that $C$ has
  $O(1)$ time complexity.

  Now let $w$ be as above and let $C$ accept $w$ in (exactly) $\tau = \tau(w)
  \le t(|w|)$ steps.
  We prove the claim by induction on $|w|$.
  The base case $|w| = n_0$ is trivial, so let $n > n_0$ and assume the claim
  holds for every word in $L$ of length strictly less than $n$.
  Consider the De Bruijn graph $G$ over the words in $|Q|^\kappa$ where $\kappa
  = 2\tau + 1$.
  Then, from the infixes of $w$ of length $\kappa$ (in order of appearance in
  $w$) one obtains a path $P$ in $G$ by starting at the leftmost infix and
  visiting every subsequent one, up to the rightmost one.
  Let $G'$ be the induced subgraph of $G$ containing exactly the nodes visited
  by $P$, and notice $P$ visits every node in $G'$ at least once.
  It is not hard to show that, for every such $P$ and $G'$, there is a path $P'$
  in $G'$ with the same starting and ending points as $P$ and that visits every
  node of $G'$ at least once while having length at most $m^2 \le
  |Q|^{2\kappa}$, where $m$ is the number of nodes in $G'$.\footnote{
    Number the nodes of $G'$ from $1$ to $m$ according to the order in which
    they are first visited by $P$.
    Then, there is a path in $G'$ from $i$ to $i+1$ for every $i \in
    \{1,\dots,m-1\}$, and a shortest such path has length at most $m$.
    Piecing these paths together along with a last (shortest) path from $m$ to
    the ending point of $P$, we obtain a path of length at most $m^2$ with the
    purported property.
  }
  To this $P'$ corresponds a word $w'$ of length $|w'| \le \kappa +
  |Q|^{2\kappa} < |Q|^{3\kappa}$ for which, by construction of $P'$ and $G'$,
  $p_{\kappa-1}(w')  = p_{\kappa-1}(w)$, $I_\kappa(w') = I_\kappa(w)$, and
  $s_{\kappa-1}(w') = s_{\kappa-1}(w')$.
  Since $\kappa \le k(|w|)$, using ($\star$) we have $|w'| < |w|$, and then
  either $|w'| \le n_0$ and $\kappa < n_0$ (since otherwise $w = w'$,
  which contradicts $|w'| < |w|$), or we may apply the induction hypothesis; in
  either case, the claim follows.
\end{proof}

\subsection{Relation to Parallel Complexity Classes}

In this section, we relate $\ACA(t)$ to other classes which characterize
parallel computation, namely the $\SC$ and (uniform) $\AC$ hierarchies.
In this context, $\SC^k$ is the class of problems decidable by Turing machines
in $O((\log n)^k)$ space and polynomial time, whereas $\AC^k$ is that decidable
by Boolean circuits with polynomial size, $O((\log n)^k)$ depth, and gates with
unbounded fan-in.
$\SC$ (resp., $\AC$) is the union of all $\SC^k$ (resp., $\AC^k$) for $k \in
\N_0$.
Here, we consider only uniform versions of $\AC$; when relevant, we state the
respective uniformity condition.
Although $\SC^1 = \L \subseteq \AC^1$ is known, it is unclear whether any other
containment holds between $\SC$ and $\AC$.

One should not expect to include $\SC$ or $\AC$ in $\ACA(t)$ for any $t \in
o(n)$.
Conceptually speaking, whereas the models of $\SC$ and $\AC$ are capable of
random access to their input, ACAs are inherently local (as evinced by
Lemmas~\ref{lem_locality_ACA}~and~\ref{lem_locality_ACA_iff}).
Explicit counterexamples may be found among the unary languages:
For any fixed $m \in \N_+$ and $w_1, w_2 \in \{ 1 \}^+$ with $|w_1|, |w_2| \ge
m$, trivially $p_{m-1}(w_1) = p_{m-1}(w_2)$, $I_m(w_1) = I_m(w_2)$, and
$s_{m-1}(w_1) = s_{m-1}(w_2)$ hold.
Hence, by Lemma~\ref{lem_locality_ACA}, if an ACA $C$ accepts $w \in \{ 1
\}^+$ in $t \in o(n)$ time and $|w|$ is large (e.g., $|w| > 4 t(|w|)$), then $C$
accepts any $w' \in \{ 1 \}^+$ with $|w'| \ge |w|$.
Thus, extending a result from \cite{sommerhalder_ACA}:

\begin{proposition} \label{prop_ACA_unary}
  If $t \in o(n)$ and $L \in \ACA(t)$ is a unary language (i.e., $L \subseteq
\Sigma^+$ and $|\Sigma| = 1$), then $L$ is either finite or co-finite.
\end{proposition}

In light of the above, the rest of this section is concerned with the converse
type of inclusion (i.e., of $\ACA(t)$ in the $\SC$ or $\AC$ hierarchies).
For $f,s,t \colon \N_+ \to \N_0$ with $f(n) \le s(n)$, we say $f$ is
\emph{constructible (by a Turing machine) in $s(n)$ space and $t(n)$ time} if
there is a Turing machine $T$ which, on input $1^n$, outputs $f(n)$ in binary
using at most $s(n)$ space and $t(n)$ time.
Also, recall a Turing machine can simulate $\tau$ steps of a CA with
$m$ (active) cells in $O(m)$ space and $O(\tau m)$ time.

\begin{proposition} \label{prop_TM_sim_ACA}
  Let $C$ be an ACA with time complexity bounded by $t \in o(n)$, $t(n) \ge
\log n$, and let $t$ be constructible in $t(n)$ space and $\poly(n)$ time.
  Then, there is a Turing machine which decides $L(C)$ in $O(t(n))$ space and
$\poly(n)$ time.
\end{proposition}

\myomit{
\begin{proof}
  We construct a machine $T$ with the purported property.
  Given an input $w$, $T$ first determines the input length $|w|$ and computes
the value $t(|w|)$ in time bounded by a polynomial $p\colon \N_+ \to \N_0$,
thus requiring $O(t(|w|) + \log n) = O(t(|w|))$ space and $O(|w| + p(|w|))$
time.
  The rest of the computation of $T$ takes place in stages, where stage $i$
corresponds to step $i$ of $C$.
  $T$ maintains a counter for $i$, incrementing it after each stage, and rejects
whenever $i > t(|w|)$.
  In stage $i$, $T$ iterates over every (active) cell $z \in \{ 0, \dots, |w| -
1 \}$ of $C$ and, by iteratively applying the local transition function of $C$
on $N_i(z)$, determines the state of $z$ in step $i$ of $C$; since $|N_i(z)| =
2i + 1$, this requires $O(i)$ space and $O(i^2)$ time for each $z$.
  If all states computed in the same stage are accept states, then $T$ accepts.
  Since $T$ runs for at most $t(|w|)$ stages, it uses $O(t(|w|))$ space and
$O(|w| \cdot t(|w|)^3 + p(|w|))$ time in total.
\end{proof}
}

Thus, for polylogarithmic $t$ (where the strict inclusion is due to
Proposition~\ref{prop_ACA_unary}):

\begin{corollary} \label{cor_ACA_SC}
  For $k \in \N_+$, $\ACA(O((\log n)^k)) \subsetneq \SC^k$.
\end{corollary}

Moving on to the $\AC$ classes, we employ some notions from descriptive
complexity theory (see, e.g., \cite{immerman_book} for an introduction).
Let $\FOL[t]$ be the class of languages describable by first-order formulas with
numeric relations in $\L$ (i.e., logarithmic space) and quantifier block
iterations bounded by $t\colon \N_+ \to \N_0$.
\myomit{%
For $C$, $Q$, $\delta$, and $\Delta$ as in Definition~\ref{def_CA}, we extend
$\Delta$ so $\Delta(c)$ is also defined for $c \in Q^{2\tau+1}$ with $\tau \in
\N_+$ by $\Delta(c)(i) = \delta(c(i), c(i+1), c(i+2))$, where $i \in \{ 0,
\dots, 2\tau - 1 \}$; in particular, $\Delta(c) \in Q^{2\tau-1}$.
Additionally, for $s \in Q$, let $\DELTA_C(c,s)$ be the relation which is true
if and only if $\Delta^\tau(c) = s$.
Note $\DELTA_C$ is computable by a Turing machine in $O(\tau)$ space and
$O(\tau^2)$ time.
}

\begin{theorem} \label{thm_ACA_FO}
  Let $t\colon \N_+ \to \N_0$ with $t(n) \ge \log n$ be constructible in
logarithmic space (and arbitrary time).
  For any ACA $C$ whose time complexity is bounded by $t$, $L(C) \in
\FOL[O(\frac{t}{\log n})]$.
\end{theorem}

\myomit{
In the following proof, we let \enquote{$\doteq$} denote the equality relation
inside a formula.

\begin{proof}
  Let $Q$ be the state set of $C$ and let $A \subseteq Q$ be the set of
accepting states of $C$; without restriction, we may assume $|Q| \ge 2$.
  In addition, let $w$ be an input for $C$ and $r = \log_{|Q|} |w|$.
  Assuming we have a predicate $\STATE_{C,w}(z, s, t')$ which is true if and
only if cell $z$ of $C$, given input $w$, is in state $s$ after $t'$ steps,
we may express whether $C$ accepts $w$ or not by the following formula:
  \[
    \varphi_{C,w} = (\exists t' \le t(|w|))
      (\forall z)
      (\exists s \in A)
      \; \STATE_{C,w}(z, s, t').
  \]
  Note $\varphi_{C,w}$ is true if and only if the initial configuration $c_0 =
c_0(w)$ of $w$ is such that an $A$-final configuration is reached in at most $t'
\le t(|w|)$ steps of $C$ (i.e., $C$ accepts $w$ in (at most) $t'$ steps).
  In turn, $\STATE_{C,w}$ may be expressed as follows:
  \begin{align*}
    \STATE_{C,w}(z, s, t') = {} &(t' \le r \land \DELTA_C(c_0|N_{t'}(z), s)) \\
        &\lor (\exists r' \le r)
      [B]^{\lfloor t' / r \rfloor}
      \; c \doteq c_0|{N_{r'}(z)}
  \end{align*}
  where $B$ is a quantifier block (iterated $\lfloor t' / r \rfloor$ times) in
which the state $s$ of $z$ is traced back $r$ steps to a former subconfiguration
$c$ (of size $2r+1$), followed by checking that the state of every cell $z'$ in
$c$ is consistent (with the computation of $C$):
  \[
    B = (\exists c \in Q^{2r + 1} . \DELTA_C(c, s))
        (\forall z' . |z - z'| \le r)
        (\exists s \doteq c(z' - z + r))
        (\exists z \doteq z').
  \]
  Note all numeric predicates in $\varphi_{C,w}$ are computable in logarithmic
space, including $\DELTA_C$.
  Since $w \in L(C)$ if and only if $w \models \varphi_{C,w}$ holds, the claim
follows.
\end{proof}
}

Since $\FOL[O((\log n)^k)]$ equals $\L$-uniform $\AC^k$ \cite{immerman_book},
by Proposition~\ref{prop_ACA_unary} we have:

\begin{corollary} \label{cor_ACA_AC}
  For $k \in \N_+$, $\ACA(O((\log n)^k)) \subsetneq \text{$\L$-uniform
$\AC^{k-1}$}$.
\end{corollary}

Because $\SC^1 \not\subseteq \AC^0$ (regardless of non-uniformity)
\cite{AC_parity}, this is an improvement on Corollary~\ref{cor_ACA_SC} at least
for $k = 1$.
Nevertheless, note the usual uniformity condition for $\AC^0$ is not
$\L$- but the more restrictive $\DLOGTIME$-uniformity \cite{vollmer_book}, and
there is good evidence that these two versions of $\AC^0$ are distinct
\cite{logspace_dlogtime_uniform_ac0}.
Using methods from \cite{Barrington_circuits}, Corollary~\ref{cor_ACA_AC} may
be rephrased for $\AC^0$ in terms of $\TIME(\polylog(n))$- or even
$\TIME((\log n)^2)$-uniformity\myomit{ (since $\DELTA_C$ is computable in
quadratic time by a Turing machine)}, but the $\DLOGTIME$-uniformity case
remains unclear.

\section{Decider ACA} \label{sec_DACA}

So far, we have considered ACAs strictly as language acceptors.
As such, their time complexity for inputs not in the target language (i.e.,
those which are not accepted) is entirely disregarded.
In this section, we investigate ACAs as \emph{deciders}, that is, as machines
which must also (explicitly) reject invalid inputs.
We analyze the case in which these decider ACAs must reject under the same
condition as acceptance (i.e., all cells are simultaneously in a final rejecting
state):

\begin{definition}[DACA] \label{def_DACA}
  A \emph{decider ACA} (\emph{DACA}) is an ACA $C$ which, in addition to its set
$A$ of accept states, has a non-empty subset $R \subseteq Q \setminus \{ q \}$
of \emph{reject states} that is disjoint from $A$ (i.e., $A \cap R =
\varnothing$).
  Every input $w \in \Sigma^+$ of $C$ must lead to an $A$- or an $R$-final
configuration (or both).
  $C$ \emph{accepts} $w$ if it leads to an $A$-final configuration $c_A$
\emph{and none of the configurations prior to $c_A$ are $R$-final}.
  Similarly, $C$ \emph{rejects} $w$ if it leads to an $R$-final configuration
$c_R$ and none of the configurations prior to $c_R$ are $A$-final.
  The time complexity of $C$ (with respect to $w$) is the number of steps
elapsed until $C$ reaches an $R$- or $A$-final configuration (for the first
time).
  $\DACA(t)$ is the DACA analogue of $\ACA(t)$.
\end{definition}

In contrast to Definition~\ref{def_ACA}, here we must be careful so that the
accept and reject results do not overlap (i.e., a word cannot be both accepted
and rejected).
We opt for interpreting the first (chronologically speaking) of the final
configurations as the machine's response.
Since the outcome of the computation is then irrelevant regardless of any
subsequent configurations (whether they are final or not), this is equivalent to
requiring, for instance, that the DACA must halt once a final configuration is
reached.

One peculiar consequence of Definition~\ref{def_DACA} is the relation between
languages which can be recognized by acceptor ACAs and DACAs (i.e., the classes
$\ACA(t)$ and $\DACA(t)$).
As it turns out, the situation is quite different from what is usually expected
of restricting an acceptor model to a decider one, that is, that deciders yield
a (possibly strictly) more restricted class of machines.
In fact, one can show $\DACA(t) \not\subseteq \ACA(t)$ holds for $t \in o(n)$
since $\SOMEONE \not\in \ACA(O(1))$ (see discussion after
Lemma~\ref{lem_locality_ACA}); nevertheless, $\SOMEONE \in \DACA(O(1))$.
For example, the local transition function $\delta$ of the DACA can be chosen as
$\delta(z_1, 0, z_2) = r$ and $\delta(z_1, z, z_2) = a$ for $z \in \{1,a,r\}$,
where $z_1$ and $z_2$ are arbitrary states, and $a$ and $r$ are the (only)
accept and reject states, respectively; see Figure~\ref{fig_dfa_one}.
Choosing the same $\delta$ for an (acceptor) ACA does \emph{not} yield an ACA
for $\SOMEONE$ since then all words of the form $0^+$ are accepted in the second
step (as they are not rejected in the first one).
We stress this rather counterintuitive phenomenon occurs only in the case of
sublinear time (as $\ACA(t) = \CA(t) = \DACA(t)$ for $t \in \Omega(n)$).

\begin{figure}[t]
  \centering
  \begin{tikzpicture}[every node/.style={block},
      block/.style={
        minimum width=width("$0$") + 3mm,
        minimum height=height("$0$") + 3mm,
        text height=1.5ex,
        text depth=.3ex,
        outer sep=0pt,
        draw, rectangle, node distance=0pt}]

    \setlength{\rowspc}{1.55em}
    \setlength{\colspc}{19.65em}

    \node    (z0)                            {$q$};
    \node    (z1)  [right=of z0]             {$0$};
    \node    (z2)  [right=of z1]             {$0$};
    \node    (z3)  [right=of z2]             {$0$};
    \node    (z4)  [right=of z3]             {$0$};
    \node    (z5)  [right=of z4]             {$0$};
    \node    (z6)  [right=of z5]             {$0$};
    \node    (z7)  [right=of z6]             {$q$};
    \draw    (z0.north west) -- ++(-0.4cm,0)
             (z0.south west) -- ++(-0.4cm,0)
             (z7.north east) -- ++(0.4cm,0)
             (z7.south east) -- ++(0.4cm,0);

    \node    (x0)  [yshift=-\rowspc]          {$q$};
    \node    (x1)  [right=of x0]             {$r$};
    \node    (x2)  [right=of x1]             {$r$};
    \node    (x3)  [right=of x2]             {$r$};
    \node    (x4)  [right=of x3]             {$r$};
    \node    (x5)  [right=of x4]             {$r$};
    \node    (x6)  [right=of x5]             {$r$};
    \node    (x7)  [right=of x6]             {$q$};
    \draw    (x0.north west) -- ++(-0.4cm,0)
             (x0.south west) -- ++(-0.4cm,0)
             (x7.north east) -- ++(0.4cm,0)
             (x7.south east) -- ++(0.4cm,0);

    \node    (w0)  [xshift=\colspc]          {$q$};
    \node    (w1)  [right=of w0]             {$0$};
    \node    (w2)  [right=of w1]             {$0$};
    \node    (w3)  [right=of w2]             {$1$};
    \node    (w4)  [right=of w3]             {$0$};
    \node    (w5)  [right=of w4]             {$1$};
    \node    (w6)  [right=of w5]             {$0$};
    \node    (w7)  [right=of w6]             {$q$};
    \draw    (w0.north west) -- ++(-0.4cm,0)
             (w0.south west) -- ++(-0.4cm,0)
             (w7.north east) -- ++(0.4cm,0)
             (w7.south east) -- ++(0.4cm,0);

    \node    (y0)  [xshift=\colspc,yshift=-\rowspc] {$q$};
    \node    (y1)  [right=of y0]             {$r$};
    \node    (y2)  [right=of y1]             {$r$};
    \node    (y3)  [right=of y2]             {$a$};
    \node    (y4)  [right=of y3]             {$r$};
    \node    (y5)  [right=of y4]             {$a$};
    \node    (y6)  [right=of y5]             {$r$};
    \node    (y7)  [right=of y6]             {$q$};
    \draw    (y0.north west) -- ++(-0.4cm,0)
             (y0.south west) -- ++(-0.4cm,0)
             (y7.north east) -- ++(0.4cm,0)
             (y7.south east) -- ++(0.4cm,0);

    \node    (y0)  [xshift=\colspc,yshift=-2\rowspc] {$q$};
    \node    (y1)  [right=of y0]             {$a$};
    \node    (y2)  [right=of y1]             {$a$};
    \node    (y3)  [right=of y2]             {$a$};
    \node    (y4)  [right=of y3]             {$a$};
    \node    (y5)  [right=of y4]             {$a$};
    \node    (y6)  [right=of y5]             {$a$};
    \node    (y7)  [right=of y6]             {$q$};
    \draw    (y0.north west) -- ++(-0.4cm,0)
             (y0.south west) -- ++(-0.4cm,0)
             (y7.north east) -- ++(0.4cm,0)
             (y7.south east) -- ++(0.4cm,0);

    \node (cm) [draw=none, right=of y7, xshift=6mm] {\cmark};
    \node (x) [draw=none, right=of x7, xshift=6mm] {\xmark};

  \end{tikzpicture}
  \caption{Computation of a DACA $C$ which decides $\SOMEONE$.
    The inputs words are $000000 \in L(C)$ and $001010 \not\in L(C)$,
    respectively.}
  \label{fig_dfa_one}
\end{figure}

Similar to (acceptor) ACAs (Lemma~\ref{lem_locality_ACA}), sublinear-time DACAs
operate locally:

\begin{lemma} \label{lem_locality_DACA}
  Let $C$ be a DACA and let $w \in \binalph^+$ be a word which $C$ decides in
exactly $\tau = \tau(w)$ steps.
  Then, for every word $w' \in \binalph^+$ with $p_{2\tau}(w) = p_{2\tau}(w')$,
$I_{2\tau + 1}(w') = I_{2\tau + 1}(w)$, and $s_{2\tau}(w) = s_{2\tau}(w')$, $C$
decides $w'$ in $\le \tau$ steps, and $w \in L(C)$ holds if and only if $w' \in
L(C)$.
\end{lemma}

\myomit{
\begin{proof}
  Let $A$ and $R$ be the set of accept and reject states of $C$, respectively,
and let $c_0$ and $c_0'$ denote the initial configurations for $w$ and $w'$,
respectively.
  Given $w \in L(C)$ (resp., $w \not\in L(C)$), it suffices to prove that, on
input $w'$, the following holds: (A) $c_\tau' = \Delta^\tau(c_0')$ is $A$-final
(resp., $R$-final); and (B) $c_{\tau'}' = \Delta^{\tau'}(c_0')$ is neither $A$-
nor $R$-final for $\tau' < \tau$.

  The proof of (A) is essentially the same as that of
Lemma~\ref{lem_locality_ACA}.
  For (B), it suffices to prove that, in every such $c_{\tau'}'$, there is a
cell $z_1$ which is not accepting as well as a cell $z_2$ which is not
rejecting.
  Since the trace of $C$ for $w$ is such that only the last configuration
$c_\tau$ can be $A$- or $R$-final, there is $z_1$ which is not accepting as well
as $z_2$ which is not rejecting in $\Delta^{\tau'}(c_0)$.
  An argument as in the proof of Lemma~\ref{lem_locality_ACA} yields this is
also the case for $c_{\tau'}'$.
\end{proof}
}

One might be tempted to relax the requirements above to $I_{2\tau + 1}(w')
\subseteq I_{2\tau + 1}(w)$ (as in Lemma~\ref{lem_locality_ACA}).
We stress, however, the equality $I_{2\tau + 1}(w) = I_{2\tau + 1}(w')$ is
crucial; otherwise, it might be the case that $C$ takes strictly less than
$\tau$ steps to decide $w'$ and, hence, $w \in L(C)$ may not be equivalent to
$w' \in L(C)$.

We note that, in addition to
Lemmas~\ref{lem_locality_ACA}~and~\ref{lem_locality_ACA_iff}, the results from
Section~\ref{sec_main_results} are extendable to decider ACAs; a more systematic
treatment is left as a topic for future work.
The remainder of this section is concerned with characterizing $\DACA(O(1))$
computation (as a parallel to Theorem~\ref{thm_ACA_SLT}) as well as establishing
the time threshold for DACAs to decide languages other than those in
$\DACA(O(1))$ (as Theorem~\ref{thm_ACA_olog} and the result $\BIN \in
\ACA(O(\log n))$ do for acceptor ACAs).

\subsection{The Constant-Time Case}

First notice that, for any DACA $C$, swapping the accept and reject states
yields a DACA with the same time complexity and which decides the complement of
$L(C)$.
Hence, in contrast to ACAs (see discussion following
Lemma~\ref{lem_locality_ACA}):

\begin{proposition} \label{prop_aca_closed_complement}
  For any $t\colon \N_+ \to \N_+$, $\DACA(t)$ is closed under complement.
\end{proposition}

Using this, we can prove the following, which characterizes constant-time DACA
computation as a parallel to Theorem~\ref{thm_ACA_SLT}:

\begin{theorem} \label{thm_DACA_LT}
  $\DACA(O(1)) = \LT$.
\end{theorem}

Hence, we obtain the rather surprising inclusion $\ACA(O(1)) \subsetneq
\DACA(O(1))$, that is, for constant time, DACAs constitute a \emph{strictly
more} powerful model than their acceptor counterparts.

\myomit{
\begin{proof}
  The inclusion $\DACA(O(1)) \subseteq \LT$ is obtained by using the locally
testable property (see Definition~\ref{def_LT}) together with
Lemma~\ref{lem_locality_DACA}.
  For the converse inclusion, we use that, for every set $F$ in the Boolean
closure over a collection $\mathcal{S}$, there are $S_1, \dots, S_m \in
\mathcal{S}$ such that:%
  \footnote{This is due to Hausdorff \cite{Hausdorff_Mengenlehre}.
    Interestingly, the same resource found application, for instance, not only
in complexity theory \cite{Buss_ttpNP, BHI} but also in characterizing related
variants of cellular automata \cite{XCA}.}
  \[
    F = S_1 \cup (\overline{S_2} \cap (S_3 \cup (\overline{S_4} \cap \cdots
(S_{m-1} \cup \overline{S_m}))))
  \]
  where $\overline{S}$ denotes the complement of $S \in \mathcal{S}$.
  Using that $\LT$ is the Boolean closure of $\mathcal{S} = \SLT$, it suffices
to show that, for any such $F$, there is a DACA $C$ with $L(C) = F$.

  First, we prove $\ACA(O(1)) \subseteq \DACA(O(1))$; due to
Theorem~\ref{thm_ACA_SLT}, $\SLT \subseteq \DACA(O(1))$ follows.
  Let $C'$ be an ACA with time complexity bounded by $t \in \N_0$.
  Since $C'$ must accept any input $w$ in at most $t$ steps, if $t + 1$ steps
elapse without $C'$ accepting, then necessarily $w \not\in L(C')$.
  Hence, $C'$ can be transformed into a DACA $C''$ by having all cells
unconditionally become rejecting in step $t + 1$.
  We have then $L(C') = L(C'')$, and $\ACA(O(1)) \subseteq \DACA(O(1))$ follows.

  For the construction of $C$, fix $F$ and use induction on the number $r \in
\N_0$ of set operations in $F$.
  The induction basis is given by $\SLT \subseteq \DACA(O(1))$.
  For the induction step, let $F' = \overline{S_2} \cap (S_3 \cup
(\overline{S_4} \cap \cdots))$ and note the complement of $F'$ is $F'' = S_2
\cup (\overline{S_3} \cap (S_4 \cup \cdots))$; hence, applying the induction
hypothesis we obtain a DACA for $F''$ and, using
Proposition~\ref{prop_aca_closed_complement}, we also obtain a DACA $C'$ with
$L(C') = F'$.
  In addition, let $C_1$ be an ACA for $S_1$ with time complexity bounded by $t
\in \N_0$.
  Now we may describe the operation of $C$:
  $C$ simulates $C_1$ on its input $w$ (while saving $w$) for $t$ steps and
accepts if $C_1$ does; otherwise, if $t + 1$ steps elapse without $C$ accepting,
then it simulates $C'$ (including its acceptance or rejection behavior).
  Thus, if $w \in S_1$, then  $w \in L(C)$; otherwise, $w \in L(C)$ if and only
if $w \in F'$.
  Since $C_1$ and $C'$ both run in constant time, so does $C$.
\end{proof}
}

\subsection{Beyond Constant Time} \label{sec_DACA_beyond_constant}

Theorem~\ref{thm_ACA_olog} establishes a logarithmic time threshold for
(acceptor) ACAs to recognize languages not in $\ACA(O(1))$.
We now turn to obtaining a similar result for DACAs.
As it turns out, in this case the bound is considerably larger:

\begin{theorem} \label{thm_DACA_sqrt}
  For any $t \in o(\sqrt{n})$, $\DACA(t) = \DACA(O(1))$.
\end{theorem}

One immediate implication is that $\DACA(t)$ and $\ACA(t)$ are
\emph{incomparable} for $t \in o(\sqrt{n}) \cap \omega(1)$ (since, e.g., $\BIN
\in \ACA(\log n)$; see Section~\ref{sec_first_obs}).
The proof idea is that any DACA whose time complexity is not constant admits an
infinite sequence of words with increasing time complexity; however, the time
complexity of each such word can be traced back to a critical set of cells which
prevent the automaton from either accepting or rejecting.
By contracting the words while keeping the extended neighborhoods of these cells
intact, we obtain a new infinite sequence of words which the DACA necessarily
takes $\Omega(\sqrt{n})$ time to decide:

\begin{proof}
  Let $C$ be a DACA with time complexity bounded by $t$ and assume $t \not\in
O(1)$; we show $t \in \Omega(\sqrt{n})$.
  Since $t \not\in O(1)$, for every $i \in \N_0$ there is a $w_i$ such that $C$
takes strictly more than $i$ steps to decide $w_i$.
  In particular, when $C$ receives $w_i$ as input, there are cells $x^i_j$ and
$y^i_j$ for $j \in \{ 0, \dots, i \}$ such that $x^i_j$ (resp., $y^i_j$) is not
accepting (resp., rejecting) in step $j$.
  Let $J_i$ be the set of all $z \in \{ 0, \dots, |w_i| - 1 \}$ for which
$\min\{|z - x^i_j|, |z - y^i_j|\} \le j$, that is, $z \in N_j(x^i_j) \cup
N_j(y^i_j)$ for some $j$.
  Consider the restriction $w_i'$ of $w_i$ to the symbols having index in
$J_i$, that is, $w_i'(k) = w_i(j_k)$ for $J_i = \{ j_0, \dots, j_{m-1} \}$ and
$j_0 < \dots < j_{m-1}$, and notice $w_i'$ has the same property as $w_i$ (i.e.,
$C$ takes strictly more than $i$ steps to decide $w_i$).
  Since $|w_i'| = |J_i| \le 2(i+1)^2$, $C$ has $\Omega(\sqrt{n})$ time
complexity on the (infinite) set $\{ w_i' \mid i \in \N_0 \}$.
\end{proof}

Using $\IDMAT$ (see Section~\ref{sec_first_obs}), we show the bound
in Theorem~\ref{thm_DACA_sqrt} is optimal:

\begin{proposition} \label{prop_DACA_IDMAT}
  $\IDMAT \in \DACA(O(\sqrt{n}))$.
\end{proposition}

We have $\IDMAT \in \ACA(O(\sqrt{n}))$ (see Section~\ref{sec_first_obs}); the
non-trivial part is ensuring the DACA also rejects every $w \not\in \IDMAT$ in
$O(\sqrt{|w|})$ time.
In particular, in such strings the $\#$ delimiters may be an arbitrary number
of cells apart or even absent altogether; hence, naively comparing every pair
of blocks is not an option.
Rather, we check the existence of a particular set of substrings of increasing
length and which must present if the input is in $\IDMAT$.
Every $O(1)$ steps the existence of a different substring is verified; the
result is that the input length must be at least quadratic in the length of the
last substring tested (and the input is timely rejected if it does not contain
any one of the required substrings).

\myomit{
\begin{proof}
  We construct an ACA $A$ which, given an input $w$, decides $\IDMAT$ in
$O(\sqrt{|w|})$ time.

  As a warm-up, first consider the case in which every block $B$ has the same
length $b \in O(\sqrt{|w|})$ and that every neighboring pair of blocks is
separated by a single $\#$.
  The leftmost cell in $B$ creates a special marker symbol $m$.
  During this first procedure, every cell which does not contain such an $m$ is
rejecting.
  At each step, if $m$ is on a cell containing a $0$ or it determines the string
it has read so far does not satisfy the regular expression $0^\ast 1 0^\ast$,
then it marks the cell as rejecting; otherwise, it does nothing (i.e., the cell
remains not rejecting).
  $m$ propagates itself to the right with speed $1$, the result being that $A$
does not reject in $i > 0$ steps if and only if for every $p \in \{ 1, 0 1,
\dots, 0^{i-1} 1 \}$ there is a block in $w$ for which $p$ is a prefix.
  It follows that $|w| \ge \frac{1}{2}i(i+1)$ and, in particular, if $b$ steps
have elapsed, then $|w| > \frac{1}{2}b^2$.
  Thus, if $A$ rejects a word during this procedure, then it does so in
$O(\sqrt{|w|})$ time.
  Once $m$ encounters $\#$, it triggers a block comparison procedure as in the
ACA $A'$ which accepts $\IDMAT$.
  This requires $O(b)$ time.
  If a violation is detected, $B$ becomes rejecting and maintains that state.
  Finally, if a number of steps have elapsed such that $A'$ would already have
accepted (which by construction of $A'$ can be determined in $O(b)$ time as a
function of $b$), $B$ becomes rejecting and maintains that result, even if it
contains cells which had been previously marked as accepting.
  Thus, $A$ accepts if and only if $A'$ does (and rejects otherwise).

  For the general case in which the block lengths vary, we let the two
procedures run in parallel, with the cells of $A$ switching between the two
back and forth.
  More precisely, the computation of $A$ is subdivided into rounds, with each
round consisting of two phases, both taking constant time each; the time
complexity of $A$, then, is directly proportional to the number of rounds
elapsed until a final configuration is reached.
  The two phases correspond to the two aforementioned procedures, that is, phase
one ($P_1$) checks that the blocks satisfy the regular expression $0^\ast 1
0^\ast$ as well as ensures the presence of the $0^\ast 1$ prefixes; phase two
($P_2$) checks the blocks are of the same length and have valid contents.
  $P_1$ advances its procedure one step at a time, while $P_2$ advances two
steps (as in the ACA construction; see Section~\ref{sec_first_obs}), and we
separate the two so as to not interfere with each other; namely, if a cell is
accepting (resp., rejecting) in one of the two phases, then it is not
necessarily so in the other one; it is only so if the procedure corresponding to
the latter phase mandates it to be so.
  If $w \in \IDMAT$, the two phases end simultaneously after $3b \in
O(\sqrt{|w|})$ steps, and $A$ accepts.
  Conversely, if $A$ rejects, then it also does so in at most $3b' \in
O(\sqrt{|w|})$ steps where $b' \in \N_+$ is maximal such that $\# 1, \# 01,
\dots, \# 0^{b'-1} 1$ are all substrings of $w \not\in \IDMAT$.
\end{proof}
}

\section{Conclusion and Open Problems} \label{sec_conclusion}

Following the definition of ACAs in Section~\ref{sec_def},
Section~\ref{sec_first_obs} reviewed existing results on $\ACA(t)$ for sublinear
$t$ (i.e., $t \in o(n)$); we also observed that sublinear-time ACAs operate in
an inherently local manner (Lemmas~\ref{lem_locality_ACA}
and~\ref{lem_locality_ACA_iff}).
In Section~\ref{sec_main_results}, we proved a time hierarchy theorem
(Theorem~\ref{thm_ACA_time_hierarchy}), narrowed down the languages in $\ACA(t)
\cap \REG$ (Theorem~\ref{thm_ACA_cap_reg}), improved Theorem~\ref{thm_ACA_log}
to $\ACA(o(\log n)) = \ACA(O(1))$ (Theorem~\ref{thm_ACA_olog}), and, finally,
obtained (strict) inclusions in the parallel computation classes $\SC$ and $\AC$
(Corollaries~\ref{cor_ACA_SC} and~\ref{cor_ACA_AC}, respectively).
The existence of a hierarchy theorem for ACAs is of interest because obtaining
an equivalent result for $\NC$ and $\AC$ is an open problem in computational
complexity theory.
Also of note is that the proof of Theorem~\ref{thm_ACA_time_hierarchy} does not
rely on diagonalization (the prevalent technique for most computational models)
but, rather, on a quintessential property of sublinear-time ACA computation
(i.e., locality as in the sense of Lemma~\ref{lem_locality_ACA}).

In Section~\ref{sec_DACA}, we considered a plausible definition of ACAs as
language deciders as opposed to simply acceptors, obtaining DACAs.
The respective constant-time class is $\LT$ (Theorem~\ref{thm_DACA_LT}), which
surprisingly is a (strict) superset of $\ACA(O(1)) = \SLT_\lor$.
Meanwhile, $\Omega(\sqrt{n})$ is the time complexity threshold for deciding
languages other than those in $\LT$ (Theorem~\ref{thm_DACA_sqrt} and
Proposition~\ref{prop_DACA_IDMAT}).

As for future work, the primary concern is extending the results of
Section~\ref{sec_main_results} to DACAs.
$\DACA(O(1)) = \LT$ is closed under union and intersection and we saw that
$\DACA(t)$ is closed under complement for any $t \in o(n)$; a further question
would be whether $\DACA(t)$ is also closed under union and intersection.
Finally, we have $\ACA(O(1)) \subsetneq \DACA(O(1))$, $\ACA(O(n)) = \CA(O(n)) =
\DACA(O(n))$, and that $\ACA(t)$ and $\DACA(t)$ are incomparable for $t \in
o(\sqrt{n}) \cap \omega(1)$; it remains open what the relation between the two
classes is for $t \in \Omega(\sqrt{n}) \cap o(n)$.

\subsubsection*{Acknowledgments.}
  I would like to thank Thomas Worsch for the fruitful discussions and feedback
during the development of this work.
  I would also like to thank the DLT 2020 reviewers for their valuable comments
and suggestions and, in particular, one of the reviewers for pointing out a
proof idea for Theorem~\ref{thm_ACA_olog}, which was listed as an open problem
in a preliminary version of the paper.

\printbibliography

\end{document}